\documentclass[12pt,onecolumn]{IEEEtran}
\usepackage{amsmath,amsthm,amssymb,bm,bbm,dsfont}
\usepackage[mathscr]{eucal}
\usepackage{upgreek}
\usepackage{setspace}
\usepackage{graphicx}
\usepackage{verbatim}
\usepackage{float}
\usepackage{placeins}
\usepackage{array}
\usepackage{booktabs}
\usepackage{threeparttable}
\usepackage{multirow}
\usepackage{amsfonts,amssymb,dsfont}
\usepackage[abs]{overpic}

\usepackage[usenames,dvipsnames]{color}
\usepackage[hidelinks]{hyperref}
\hypersetup{
	unicode=false,          
	pdftoolbar=true,        
	pdfmenubar=true,        
	pdffitwindow=false,     
	pdfstartview={FitH},    
	pdftitle={My title},    
	pdfauthor={Author},     
	pdfsubject={Subject},   
	pdfcreator={Creator},   
	pdfproducer={Producer}, 
	pdfkeywords={keyword1} {key2} {key3}, 
	pdfnewwindow=true,      
	colorlinks=true,        
	linkcolor=Red,          
	citecolor=ForestGreen,  
	filecolor=Magenta,      
	urlcolor=BlueViolet,    
}
\usepackage{doi}
\usepackage{url}
\usepackage{enumitem}

\makeatletter
\newcommand{\opnorm}{\@ifstar\@opnorms\@opnorm}
\newcommand{\@opnorms}[1]{%
	\left|\mkern-1.5mu\left|\mkern-1.5mu\left|
	#1
	\right|\mkern-1.5mu\right|\mkern-1.5mu\right|
}
\newcommand{\@opnorm}[2][]{%
	\mathopen{#1|\mkern-1.5mu#1|\mkern-1.5mu#1|}
	#2
	\mathclose{#1|\mkern-1.5mu#1|\mkern-1.5mu#1|}
}
\makeatother

\makeatletter
\ifx\@NODS\undefined%

\let\mathbb=\mathds
\else%
\fi
\makeatother

\usepackage[backend=bibtex,
hyperref=true,
url=true,
isbn=true,
backref=false,
style=ieee,
sorting=none,
block=none]{biblatex}
\bibliography{reference.bib}

\DeclareMathOperator{\Tr}{Tr}

\newcommand{\be}{{\mathbf e}}

\def\cX{{\cal X}}
\def\cY{{\cal Y}}        

\def\0{{\mathbf{0}}}
\def\1{{\mathbf{1}}}
\def\2{{\mathbf{2}}}
\def\3{{\mathbf{3}}}
\def\4{{\mathbf{4}}}
\def\5{{\mathbf{5}}}
\def\6{{\mathbf{6}}}

\def\7{{\mathbf{7}}}
\def\8{{\mathbf{8}}}
\def\9{{\mathbf{9}}}

\def\be{\begin{equation}}
\def\ee{\end{equation}}
\def\bea{\begin{eqnarray}}
\def\eea{\end{eqnarray}}

\theoremstyle{plain}
\newtheorem{theo}{Theorem} 
\newtheorem{prop}[theo]{Proposition} 
\newtheorem{lemm}[theo]{Lemma} 

\theoremstyle{definition}

\theoremstyle{remark}

\begin{document}

\title{{ On the Concavity of Auxiliary Function in Classical-Quantum Channels}}

\author{
\IEEEauthorblockN{Hao-Chung Cheng \\}
\IEEEauthorblockA{National Taiwan University, Taiwan (R.O.C.) \&\\
University of Technology Sydney, Australia \\
Email: \texttt{\href{mailto:F99942118@ntu.edu.tw}{F99942118@ntu.edu.tw}}}
\and
\IEEEauthorblockN{\\ Min-Hsiu Hsieh \\}
\IEEEauthorblockA{
University of Technology Sydney, Australia \\
Email: \texttt{\href{mailto:Min-Hsiu.Hsieh@uts.edu.au}{Min-Hsiu.Hsieh@uts.edu.au}}}
}

\maketitle

\begin{abstract}
The auxiliary function of a classical channel appears in two fundamental quantities that upper and lower bound the error probability, respectively. A crucial property of the auxiliary function is its concavity, which leads to several important results in finite block length analysis. In this paper, we prove that the auxiliary function of a classical-quantum channel also enjoys the same concave property, extending an earlier partial result to its full generality. The key component in our proof is a beautiful result of geometric means of operators. 
\end{abstract}

\section{Introduction} \label{sec:introduction}
Denote by $\mathscr{P}(\mathcal{X})$ the set of probability distributions on a finite set $\mathcal{X} = \{1,2,\ldots, |\mathcal{X}| \}$. For any fixed $P\in\mathscr{P}(\mathcal{X})$ and $s\geq0$, the \emph{auxiliary function} $E_0(s,P)$ of a classical communication channel $Q(y|x)$ with the output set $\mathcal{Y} = \{1,2,\ldots, |\mathcal{Y}| \}$ is defined as
\begin{equation}\label{eq_cE0}
E_0(s,P) \triangleq -\log\left[\sum_{y\in\cY} \left(\sum_{x\in\cX} P(x)Q(y|x)^{\frac{1}{1+s}}\right)^{1+s}\right].
\end{equation}
This function appears in two fundamental quantities in classical information theory: for any $R\geq 0$,
\begin{align} \label{eq:random}
E_\text{r}(R) \triangleq \max_{0\leq s\leq 1} \left\{
\max_{P\in\mathscr{P}(\mathcal{X})} E_0(s,P) - sR
\right\},
\end{align}
and
\begin{align} \label{eq:sp}
E_\text{sp}(R) \triangleq \sup_{s\geq 0}\left\{
\max_{P\in\mathscr{P}(\mathcal{X})} E_0(s,P) - sR
\right\},
\end{align}
where $E_{\rm r}(R)$ is called the \emph{random coding exponent} and $E_{\rm sp}(R)$ is called the \emph{sphere-packing exponent} of the classical channel $Q$. These two quantities are critical since, for any block length $n$ and any rate $R\geq0$, the error probability $P_e(n,R)$, minimized over all  possible coding strategies, satisfies \cite{Gal68}
\begin{equation}\label{eq_cPe}
2^{-n E_{\rm sp}(R)} \leq P_e(n,R)\leq 2^{-n E_{\rm r}(R)}.
\end{equation}
Consequently, properties of the auxiliary function $E_0(s,P)$ reveal important functional behaviour of the two exponents, and lead to a deeper understanding of the error probability of a given classical channel $Q$. It is well-known (and easy to show) \cite{Gal68}: $\forall s\geq 0$,
\begin{eqnarray}
E_0(s,P) &\geq& 0; \\
\frac{\partial E_{0}(s,P)}{\partial s} &>& 0; \\
\frac{\partial^2 E_{0}(s,P)}{\partial s^2} &\leq& 0. \label{eq:concavity}
\end{eqnarray}
It turns out that $E_0(s,P)$ is concave in $s\geq 0$. In addition to other important contributions in finite block length analysis, this fact also provides an alternative proof to Shannon's noiseless channel coding theorem \cite{Sha48}.
 
In recent years, much attention has been paid to understanding the error probability of a quantum channel. In this scenario, it suffices to consider a \emph{classical-quantum channel}, which is a mapping $W:x\in\cX\mapsto W_x\in\mathcal{S(H)}$ from the finite set $\mathcal{X}$ to $\mathcal{S(H)}$, i.e., the set of density operators (positive semi-definite operators with unit trace) on a fixed Hilbert space $\mathcal{H}$. Given a classical-quantum channel $W$ and a distribution $P$ on the input $\cX$, we can similarly define the \emph{auxiliary function} $E_0(s,P)$\footnote{Here, we slightly abuse the notation since it should be clear from the context the underlining channel is quantum or classical.} \cite{BH98, Hol00}: $\forall s\geq0$,
\begin{align} \label{eq:E0}
E_0(s,P) \triangleq -\log \Tr \left[
\left( \sum_{x\in\mathcal{X}} P(x) \cdot W_x^{\frac1{1+s}}\right)^{1+s}
\right].
\end{align}
 This quantity is a quantum generalization of Eq.~(\ref{eq_cE0}), and recovers Eq.~(\ref{eq_cE0}) when all $\{W_x\}_{x\in\cX}$ commute. 
 
The auxiliary function $E_0(s,P) $ in Eq.~(\ref{eq:E0}) also appears in the random coding exponent $E_{\rm r}(R)$ and the sphere-packing exponent $E_{\rm sp}(R)$ of a classical-quantum channel $W$, which can be similarly defined as that in Eqs.~(\ref{eq:random}) and (\ref{eq:sp}), respectively. However, the relations between these two exponents and the error probability of the underlining classical-quantum channel $W$ are much harder to obtain. The random coding exponent $E_{\rm r}(R)$  is shown to be an upper bound to the error probability of a classical-quantum channel $W$ when every $W_x$ is pure (i.e. the density operator $W_x$ is a rank-one matrix) in Ref.~\cite{BH98}, and it is conjectured to hold for general quantum states. Furthermore, the sphere-packing bound that lower bounds the error probability of $W$ was recently proved in Ref.~\cite{Dal13}\footnote{However, this bound only works in the asymptotic regime $n\to\infty$, unlike the classical case in Eq.~(\ref{eq_cPe}) that holds for any $n\in\mathbb{N}$ and $R\geq0$.}. These results are highly nontrivial due to the non-commutative nature of the density operators involved in their definitions. Many important questions in quantum information theory are still left open. Notably, it is still unknown whether the auxiliary function $E_0(s,P) $ in Eq.~(\ref{eq:E0}) is concave for all $s\geq0$. This might be one reason that the error probability of any finite block length $n$ is less understood in the quantum regime.  Note that $E_0(s,P) $ has been shown to be concave in $0\leq s \leq 1$ in Ref.~\cite{FY06}. Its proof relies on an \emph{ad-hoc} operator inequality in order to show that the second-order derivative of $E_0(s,P)$ is non-positive for $s\in [0,1]$. However, this method seems impossible to work for all $s\geq 0$.

In this paper, we are able to prove that $E_0(s,P)$ of a classical-quantum channel $W$ is concave for all $s\geq 0$. Our proof culminates the latest development of operator algebra; in particular, the beautiful theory of a general geometric mean of operators \cite{KA80}. Our proof can be viewed as a direct generalization of its classical proof in Ref.~\cite[Theorem 5.6.3]{Gal68}. 

The paper is organized as follows. Sec~\ref{GM} presents the main technical tool, i.e., the ``$s$-weighted geometric mean". The main result is presented in Sec~\ref{main}, and our conclusion is given in Sec~\ref{sec:conclusion}. 

\section{Technical Tools} \label{GM}
Denote by $\mathbb{M}_d^+$ and $\mathbb{M}_d^{++}$ the set of positive semi-definite matrices and positive definite matrices, respectively. 
For two $d\times d$ Hermitian matrices $A$ and $B$, we denote by $A \succeq B$ if $A-B \in \mathbb{M}_d^+$. Let $A,B\in\mathbb{M}_d^{++}$. Then the ``$s$-weighted geometric mean" of $A$ and $B$ is defined as
\begin{align}
A\#_s B \triangleq A^{1/2} \left( A^{-1/2} B A^{-1/2} \right)^s A^{1/2}.
\end{align}
The geometric mean enjoys following properties \cite{KA80,CPR94,LL07} (see also \cite[Chapter 6]{Bha09} and \cite[Section 4]{Hia97}).
\begin{prop} 
	[Properties of Geometric Means] \label{prop:properties}
	Let $A,B,C,D\in\mathbb{M}_d^{++}$ and $s\in\mathbb{R}$. Then
	\begin{itemize}
		\item [\textnormal{(a)}] \textnormal{Commutativity:} $A\#_s B = A^{1-s} B^s$ for $AB=BA$;

		\item [\textnormal{(b)}] \textnormal{Joint homogeneity:} $(aA)\#_s(bB) = a^{1-s}b^s (A\#_s B)$ for $a,b>0$;

		\item [\textnormal{(c)}] \textnormal{Monotonicity:} $A\#_s B \preceq C\#_s D$ for $A\preceq C$, $B\preceq D$ and $s\in[0,1]$;

		\item [\textnormal{(d)}] \textnormal{Congruence invariance:} For every non-singular matrix $M$, $M(A\#_sB)M^\dagger = \left(MAM^\dagger\right) \#_s \left( MBM^\dagger \right)$;

		\item [\textnormal{(e)}] \textnormal{Self-duality:} $A \#_s B = B\#_{1-s} A$, and $ \left( A\#_s B\right)^{-1} = A^{-1} \#_s B^{-1}$;

		\item [\textnormal{(f)}] \textnormal{Concavity:} 
		\begin{align}
		\begin{split}
		&\left(\lambda A + (1-\lambda)B \right)\#_s \left(\lambda C + (1-\lambda)D \right)\\
		&\succeq
		 \lambda \left( A\#_s C\right) + (1-\lambda)\left( B\#_s D\right)
		 \end{split}
		\end{align}
		 for all $\lambda, s\in[0,1]$;

		\item [\textnormal{(g)}] \textnormal{HM-GM-AM inequality:} $\left( (1-s) A^{-1} + s B^{-1} \right)^{-1} \preceq A\#_s B \preceq (1-s) A + sB$ for $s\in[0,1]$.	 
		
		\item[\textnormal{(h)}] \textnormal{Continuity:} $A\#_s B$ is continuous in $A$ and $B$.
	\end{itemize}
\end{prop}

Let $x\triangleq (x_1,\ldots, x_d) \in \mathbb{R}^d$ be a $d$-dimensional vector. Denote by $x^{\downarrow}\triangleq (x_1^\downarrow, \ldots, x_d^\downarrow)$ the \emph{decreasing arrangement} of $x$, i.e.~$x_1^\downarrow\geq \cdots \geq x_d^\downarrow$. 
We say that $x$ is \emph{weak majorized} by $y$, denoted by $x\prec_w y$, if
\begin{align}
\sum_{j=1}^k x_j^\downarrow \leq \sum_{j=1}^k y_j^\downarrow, \; 1\leq k \leq d.
\end{align}
The \emph{weak log-majorization} is $x\prec_\textnormal{log} y$ is defined when $\log x \prec_\textnormal{log} \log y$. It is well-known that $\log x \prec_\textnormal{log} \log y$ implies $x\prec_w y$ \cite[Example II.3.5]{Bha97}.
For two Hermitian matrices $A$ and $B$, if $\lambda(A) \prec_w \lambda(B)$, then $\opnorm{A} \leq \opnorm{B}$ for any unitarily-invariant norm $\opnorm{\,\cdot\,}$ \cite[Theorem 6.18]{HJ14}.

In the following, we collect a few lemmas that will be used in the main proof.
\begin{lemm}[Matharu \& Aujla {\cite[Theorem 2.10]{MA12}}] \label{lemm:MA12}
	For any $A,B\in\mathbb{M}_d^{++}$, and $0\leq s \leq 1$. Then
	\begin{align}
	\lambda( A \#_s B) \prec_\textnormal{log} \lambda\left( A^{1-s} B^{s} \right).
	\end{align}
\end{lemm}
\begin{lemm}
	[{\cite[Theorem IX.2.10]{Bha97}}] \label{lemm:Bha1}
	Let $A,B\in\mathbb{M}_d^{+}$. Then for every unitarily-invariant norm $\opnorm{\,\cdot\,}$, we have
	\begin{align}
	&\opnorm{B^tA^tB^t} \leq \opnorm{(BAB)^t},\quad \text{for } 0\leq t\leq 1,\\
	&\opnorm{B^tA^tB^t} \geq \opnorm{(BAB)^t},\quad \text{for } t\geq 1.\\
	\end{align}
\end{lemm}
\begin{lemm}
	[{\cite[Example II.3.5]{Bha97}}] \label{lemm:Bha2}
	Let $x,y \in \mathbb{R}_+^d$ (the set of $d$-dimensional vectors of non-negative real numbers).
	Then 
	\begin{align}
	x\prec_w y \text{ implies } x^t \prec_w y^t
	\end{align}
	for all $t\geq 1$.
\end{lemm}
\begin{lemm} 
	[See, e.g.~{\cite[Section 2.2]{Car09}}] \label{lemm:convex}
	Let $f$ be a convex function on real lines. Then $A\preceq B$ implies
	\begin{align}
	\Tr\left[ f(A) \right] \leq \Tr \left[ f(B) \right].
	\end{align} 
\end{lemm}
\begin{lemm}
	[Matrix H\"older's Inequality {\cite[Corollary IV.2.6]{Bha97}}] \label{lemm:Holder}
	Let $A,B\in\mathbb{M}_d^{+}$.
	Then 
	\begin{align}
	\Tr\left[ AB \right] \leq 
	\left( \Tr\left[A^{\frac1{\theta}}\right]\right)^\theta
	\left( \Tr\left[B^{\frac1{1-\theta}}\right]\right)^{1-\theta}
	\end{align}
	for all $0\leq \theta \leq 1$.
\end{lemm}

\section{Main Result} \label{main}

We first recall a few notations. Let $\mathcal{X} = \{1,2,\ldots, |\mathcal{X}| \}$ be a finite alphabet.
Denote by $\mathscr{P}(\mathcal{X})$ the set of probability distributions on $\mathcal{X}$.
Fix a (separable) Hilbert space $\mathcal{H}$.
The set of density operators (i.e.~positive semi-definite operators with unit trace) on $\mathcal{H}$ is defined as $\mathcal{S(H)}$.
Denote the set of all classical-quantum (c-q) channels $W$ from $\mathcal{X}$ to $\mathcal{S(H)}$ by $\mathscr{W}(\mathcal{X})$.


\begin{theo} \label{theo:main}
	Given a classical-quantum channel $W\in\mathscr{W}(\mathcal{X})$ and a distribution $P\in\mathscr{P}(\mathcal{X})$,
	the auxiliary function $E_0(s,P)$ is concave in $s\geq 0$.
\end{theo}

\begin{proof} We first present the proof that only works when all $\{W_x\}_{x\in\cX}$ are full rank. The proof can then be relaxed to include the non-invertible case.

	Let $X$ be a random variable with distribution $P$, and denote by $\mathbb{E}$ the expectation with respect to $P$. Then it suffices to prove the convexity of the map:
	\begin{align} \label{eq:main1}
	f:t \mapsto \log \Tr \left[
	\left( \mathbb{E} \, W_X^{\frac1{t}}\right)^{t} \right]
	\end{align}
	for all $t\geq 1$.
	
	Before starting the proof, we first prepare the following lemma that is crucial in our derivations.
	\begin{lemm} \label{lemm:our}
		Let $A,B\in\mathbb{M}_d^{++}$. Then, for every $t \geq 1$ and $0\leq \lambda \leq 1$, we have
		\begin{align}
		\Tr\left[ \left( A\#_{\lambda}B \right)^t \right]
		\leq 
		\Tr \left[ A^{t (1-\lambda)} B^{t \lambda} \right].
		\end{align}
	\end{lemm}
	\begin{proof}
		From Lemma \ref{lemm:MA12}, we have
		\begin{align}
		\Tr \left[ A\#_{\lambda} B \right]
		&\leq 
		\Tr \left[ A^{\frac{1-\lambda}2} B^\lambda A^{\frac{1-\lambda}2} \right]\\
		&\leq \Tr \left[ \left( A^{\frac{t(1-\lambda)}2} B^{t\lambda}  A^{\frac{t(1-\lambda)}2} \right)^{\frac1t} \right],
		\end{align}
		where the last inequality follows from Lemma \ref{lemm:Bha1}.
		Next, applying Lemma \ref{lemm:Bha2} on the above inequality yields
		\begin{align}
		\Tr \left[ \left(A\#_{\lambda} B\right)^t \right]
		\leq \Tr \left[ \left( \left( A^{\frac{t(1-\lambda)}2} B^{t\lambda}  A^{\frac{t(1-\lambda)}2} \right)^{\frac1t}  \right)^t \right],
		\end{align}
		which completes the proof.
	\end{proof}
	
We now begin the proof of Theorem~\ref{theo:main}. These steps follow closely with those in Ref.~\cite[Theorem 5.6.3]{Gal68}. Let $l,r$, and $\theta$ be arbitrary numbers $1\leq l\leq r$, $0\leq \theta \leq1$, and define
	\begin{align}
	t = \theta l + (1-\theta) r.
	\end{align}
	Let $t\equiv 1+s \geq 1$. Then we prove the convexity of the map $f$, i.e.~
	\begin{align}
	f(t) \leq \theta f(l) + (1-\theta) f(r).
	\end{align}

	Define the number $\lambda$ by 
	\begin{align}
	\lambda = \frac{l\theta}{t}; \quad 1-\lambda = \frac{r(1-\theta)}{t}.
	\end{align}
	Then it follows that
	\begin{align}
	\frac1t = \frac{\theta}{t} + \frac{1-\theta}{t}
	 = \frac{\lambda}{l} + \frac{1-\lambda}{r}.
	\end{align}
	The convexity of the geometric means (see item (f) in Proposition \ref{prop:properties}) implies that
	\begin{align}
	\mathbb{E} \left[W^{1/t}\right]
	&= \mathbb{E} \left[ W^{\lambda /l } W^{(1-\lambda)/r} \right]\\
	&= \mathbb{E} \left[ W^{1 /l } \#_{1-\lambda} W^{1 /r } \right] \\
	&\preceq  \mathbb{E} \left[ W^{1 /l } \right] \#_{1-\lambda} \mathbb{E} \left[ W^{1 /r } \right].
	\end{align} 
	Now let $A\equiv \mathbb{E} \left[ W^{1 /l } \right]$ and $B\equiv\mathbb{E} \left[ W^{1 /r } \right]$.
	Since $x\mapsto x^t$ for $t\geq 1$ is a convex function, Lemma \ref{lemm:convex} leads to
	\begin{align}
	\Tr\left[\left (\mathbb{E} \left[W^{1/t}\right] \right)^t\right]
	&\leq \Tr \left[ \left( A\#_{1-\lambda} B
	\right)^t \right] \\
	&\leq \Tr \left[ A^{t \lambda} B^{t (1-\lambda) } \right] \label{eq:main2}\\
	&= \Tr \left[ A^{l \theta } B^{r (1-\theta) } \right], \label{eq:main3}
	\end{align}
	where Eq.~\eqref{eq:main2} follows from Eq.~(\ref{lemm:our}). Finally, applying matrix H\"older's inequality, Lemma \ref{lemm:Holder}, on the right-hand side of Eq.~\eqref{eq:main3}, we have
	\begin{align}
	\Tr\left[\left (\mathbb{E} \left[W^{1/t}\right] \right)^t\right]
	&\leq
	\left( \Tr\left[A^{l}\right]\right)^\theta
	\left( \Tr\left[B^{r}\right]\right)^{1-\theta}\\
	&= \left( \Tr \left( \mathbb{E}\left[ W^{1 /l }  \right] ^{l}  \right)\right)^\theta
	\left( \Tr \left( \mathbb{E}\left[ W^{1 /r } \right]^{r}\right)\right)^{1-\theta}.
	\end{align}
	Taking logarithm on the above inequality arrives at $ f(t) \leq \theta f(l) + (1-\theta) f(r)$. This completes the proof for the special case of invertible channel outputs.

The above proof assumes that every realization of the density operator $W_X$ is positive definite.
Hence, each density operator $W_x^{\lambda/l} W_x^{(1-\lambda)/r}$ can be expressed as a geometric mean $W_x^{1/l} \#_s W_x^{1/r}$.
	However, if $W_x$ is not invertible for some $x\in\cX$, then consider a sequence of positive definite operators $W_{x,\epsilon} \triangleq W_x + \epsilon I$ that approximate $W_x$, i.e.,~$\lim_{\epsilon\searrow 0} W_{x,\epsilon} = W_x$. The geometric mean of $W_x^{1/l}$ and $W_x^{1/r}$ is defined as 
	\begin{align} \label{eq:general}
	W_x^{1/l} \#_s W_x^{1/r} \triangleq \lim_{\epsilon\searrow 0} W_{x,\epsilon}^{1/l} \#_s W_{x,\epsilon}^{1/r},
	\end{align}
by the continuity of the geometric means (see item (h) in Proposition \ref{prop:properties}).  Note that the concavity of the geometric means, and Lemmas \ref{lemm:MA12} and \ref{lemm:our} still hold if we use the definition in Eq.~\eqref{eq:general}. We can thus obtain a complete proof. 

\end{proof}

\section{Conclusion} \label{sec:conclusion}
In this paper, we proved an open question that was originally raised in \cite{Hol00}. A partial result to this question was obtained in \cite{FY06}; however, we can extend the concavity of the auxiliary function $E_0(s,P)$ for all $s\geq0$. Consequently, the definition of auxiliary function (\ref{eq:E0}) of a classical-quantum channel exactly recovers its classical counterpart \cite{Gal68}, a quantity that plays a crucial role in classical information theory.  We hope that this concave property will also allow us to better characterize the error probability of a classical-quantum channel in the finite regime. 

\section*{Acknowledgements}
MH is supported by an ARC Future Fellowship under Grant FT140100574.

\printbibliography

\end{document}